\documentclass[a4paper,11pt]{article}

\usepackage{fullpage}
\usepackage{times}
\usepackage{soul}
\usepackage{url}
\usepackage[utf8]{inputenc}
\usepackage[small]{caption}
\usepackage{graphicx}
\usepackage{xcolor}
\usepackage{amsmath}
\usepackage{booktabs}
\usepackage{algorithm}
\usepackage{enumitem}

\usepackage{amsthm}
\usepackage{amssymb}
\usepackage{algpseudocode}

\usepackage{bbm}

\newtheorem{theorem}{Theorem}[section]
\newtheorem{corollary}[theorem]{Corollary}

\newtheorem{lemma}[theorem]{Lemma}

\newtheorem{claim}{Claim}

\theoremstyle{definition}

\usepackage{natbib}

\DeclareMathOperator{\val}{Val}
\DeclareMathOperator{\opt}{OPT}
\newcommand{\N}{\mathbb{N}}
\newcommand{\mbb}{\text{MBB}}
\newcommand{\ef}{\text{EF}}
\newcommand{\Mod}[1]{\ (\mathrm{mod}\ #1)}
\newcommand{\cA}{\mathcal{A}}

\allowdisplaybreaks

\usepackage{authblk}

\title{\bf On the Complexity of Fair House Allocation}

\author[1]{Naoyuki Kamiyama}
\author[2]{Pasin Manurangsi}
\author[3]{Warut Suksompong}

\affil[1]{Institute of Mathematics for Industry, Kyushu University, Japan}
\affil[2]{Google Research, USA}
\affil[3]{School of Computing, National University of Singapore, Singapore}

\date{\vspace{-10mm}}

\begin{document}

\maketitle

\begin{abstract}
We study fairness in house allocation, where $m$ houses are to be allocated among $n$ agents so that every agent receives one house.
We show that maximizing the number of envy-free agents is hard to approximate to within a factor of $n^{1-\gamma}$ for any constant $\gamma>0$, and that the exact version is NP-hard even for binary utilities.
Moreover, we prove that deciding whether a proportional allocation exists is computationally hard, whereas the corresponding problem for equitability can be solved efficiently.
\end{abstract}

\section{Introduction}

We consider the classical setting of \emph{house allocation}, also known as \emph{assignment} \citep{HyllandZe79,Zhou90,AbdulkadirogluSo98}.
In this setting, there are $m$ houses to be allocated among $n\le m$ agents, with no two agents sharing the same house.
The agents have possibly different preferences over the houses, and each agent should be assigned exactly one house.

While house allocation has typically been considered from the economic efficiency and strategyproofness perspectives \citep{AbrahamCeMa04,KrystaMaRa14}, another important concern is \emph{fairness}: it is desirable that the agents feel fairly treated.
For example, the prominent fairness notion of \emph{envy-freeness} means that agents do not envy one another with respect to their assigned houses.
When $m=n$, all of the houses must be assigned, so an agent is envy-free if and only if she receives one of her most preferred houses.
Thus, in order to compute an assignment with the largest number of envy-free agents, it suffices to find a maximum matching in the bipartite graph where the two sets of vertices correspond to the agents and the houses, respectively, and there is an edge between an agent and a house exactly when the house is among the agent's most preferred houses---it is well-known that this task can be done in polynomial time.
\citet{BeynierChGo19} assumed that agents can only envy other agents with whom they are acquainted according to a given acquaintance network, and provided algorithms and hardness results for various networks when $m=n$.
\citet{GanSuVo19} addressed the general setting with $m\ge n$ (without an acquaintance network).
In this setting, a simple matching algorithm no longer suffices, since even when all agents prefer the same house, it may still be possible to achieve envy-freeness by not allocating this house.
Gan et al.~devised a polynomial-time algorithm that decides whether an envy-free assignment exists and, if so, computes one such assignment.
However, their work left open the question of whether an assignment maximizing the number of envy-free agents can be computed efficiently---after all, when making all agents envy-free is impossible, the number of envy-free agents is a natural optimization objective.

In this note, we give a strong negative answer to the question above by showing that under well-known complexity-theoretic assumptions, perhaps surprisingly, it is hard not only to maximize the number of envy-free agents, but also to obtain any decent approximation thereof.
Specifically, assuming the \emph{Small Set Expansion Hypothesis} \citep{RaghavendraSt10}, the problem is hard to approximate to within a factor of $n^{1-\gamma}$ for any constant $\gamma > 0$.
We also establish that even when the agents have binary utilities over the houses, maximizing the number of envy-free agents is NP-hard.
In addition, we consider two other important fairness notions: proportionality and equitability.
On the one hand, we show that deciding whether a proportional allocation exists is NP-hard, thereby drawing a sharp contrast to the envy-freeness result of \citet{GanSuVo19}; on the other hand, we prove that the corresponding problem for equitability can be solved efficiently.

\section{Preliminaries}
\label{sec:prelim}

Let $[k]$ denote the set $\{1,2,\dots,k\}$ for any positive integer $k$.
In the house allocation setting, there is a set $A = \{a_1,\dots,a_n\}$ of $n$ agents and a set $H=\{h_1,\dots,h_m\}$ of $m \geq n$ houses.
Each agent $a\in A$ has a utility $u_a(h)\ge 0$ for a house $h\in H$.
The utilities are said to be \emph{binary} if $u_a(h)\in\{0,1\}$ for all $a\in A$ and $h\in H$.
As \emph{assignment} or \emph{house allocation} is an injection $\phi: A \to H$.
We consider the following fairness properties of assignments:
\begin{itemize}
\item An agent $a\in A$ is said to be \emph{envy-free} in an assignment $\phi$ if $u_a(\phi(a)) \ge u_a(\phi(a'))$ for all $a'\in A$.
For an assignment $\phi$, denote by $\val(\phi)$ the number of envy-free agents in $\phi$.
The assignment $\phi$ is called \emph{envy-free} if $\val(\phi) = n$.
\item An agent $a\in A$ is said to be \emph{proportional} in an assignment $\phi$ if $u_a(\phi(a)) \ge \frac{1}{n}\sum_{a'\in A}u_a(\phi(a'))$.
An assignment $\phi$ is called \emph{proportional} if all $n$ agents are proportional.
\item An assignment $\phi$ is called \emph{equitable} if $u_a(\phi(a)) = u_{a'}(\phi(a'))$ for all agents $a,a'\in A$.
\end{itemize}
Notice that for envy-freeness, it suffices to consider the agents' ordinal rankings over the houses, whereas the cardinal utilities play an important role in the definitions of proportionality and equitability.
All three notions are commonly studied in the unconstrained allocation setting where each agent can receive any number of items \citep{BouveretChMa16,Markakis17}. However, to the best of our knowledge, the latter two notions have not been previously studied in house allocation.

\section{Envy-Freeness}

We begin by considering envy-freeness.
Recall that \citet{GanSuVo19} gave a polynomial-time algorithm for deciding whether an envy-free assignment exists for any given instance.
We show that their algorithm cannot be generalized to efficiently compute the maximum number of envy-free agents, or even any decent approximation thereof, provided that known complexity-theoretic assumptions hold.
We refer to our problem of interest as {\normalfont \scshape Maximum Envy-Free Assignment}.

To obtain the hardness of approximation, we will reduce from the so-called {\normalfont \scshape Maximum Balanced Biclique (MBB)}  problem.
In this problem, we are given a bipartite graph $G = (L, R, E)$, and the goal is to find a balanced complete bipartite subgraph (i.e., $K_{p, p}$ for some $p \in \N$, also called a balanced \emph{biclique}) of $G$ with the maximum number of vertices. We use $\opt_{\mbb}(G)$ to denote the largest $p$ such that $K_{p, p}$ is a subgraph of $G$, and $N$ to denote the number of left vertices of $G$ (i.e., $N = |L|$).
We claim that any approximation algorithm for {\normalfont \scshape Maximum Envy-Free Assignment} can be turned into an approximation algorithm for {\normalfont \scshape MBB}, with a multiplicative loss of roughly $2$ in the approximation ratio.

\begin{theorem} \label{thm:max-ef-assignment-hardness}
For any constant $\varepsilon > 0$, if there exists a polynomial-time $f(n)$-approximation algorithm for {\normalfont \scshape Maximum Envy-Free Assignment}, then there is a polynomial-time $2(1 + \varepsilon)\cdot f(N)$-approximation algorithm for {\normalfont \scshape Maximum Balanced Biclique}.
\end{theorem}

While {\normalfont \scshape Maximum Balanced Biclique} is known to be NP-hard \citep{GareyJo79}, the NP-hardness of approximating it remains open.
Nevertheless, several inapproximability results for the problem are known under different complexity-theoretic assumptions~\citep{Feige02,Khot06,BhangaleGaHa16,Manurangsi17,Manurangsi17-2}. 
Specifically, assuming that NP cannot be solved in subexponential time (i.e., $NP \nsubseteq \bigcap_{\delta > 0} BPTIME(2^{n^{\delta}})$), our theorem together with the hardness result of \citet{Khot06} implies that {\normalfont \scshape Maximum Envy-Free Assignment} is hard to approximate to within a factor of $n^{\gamma}$ for some constant $\gamma > 0$. 
Furthermore, combining our theorem with the hardness of \citet{Manurangsi17-2}, we can deduce that {\normalfont \scshape Maximum Envy-Free Assignment} is hard to approximate to within a factor of $n^{1 - \gamma}$ for any constant $\gamma > 0$---this assumes the so-called Small Set Expansion Hypothesis~\citep{RaghavendraSt10}, which is itself a strengthening of the seminal Unique Games Conjecture~\citep{Khot02}. 
This $n^{1 - \gamma}$ inapproximability ratio nearly matches an $n$-approximation, which can be trivially achieved by ensuring that a single agent is envy-free.

\begin{corollary}
If $NP \nsubseteq \bigcap_{\delta > 0} BPTIME(2^{n^{\delta}})$, then, for some constant $\gamma > 0$, {\normalfont \scshape Maximum Envy-Free Assignment} cannot be approximated to within a factor of $n^{\gamma}$ in polynomial time. 
\end{corollary}

\begin{corollary}
If the Small Set Expansion Hypothesis holds, then {\normalfont \scshape Maximum Envy-Free Assignment} is NP-hard to approximate to within a factor of $n^{1 - \gamma}$ for any constant $\gamma > 0$. 
\end{corollary}

Our main technical contribution is the following reduction from {\normalfont \scshape Maximum Balanced Biclique} to {\normalfont \scshape Maximum Envy-Free Assignment}, as formalized below.

\begin{lemma}[Reduction] \label{lem:thm:max-ef-assignment-reduction}
There is a polynomial-time reduction that takes an instance $G = (L, R, E)$ of {\normalfont \scshape Maximum Balanced Biclique} and produces an instance $(A, H, \{u_a\}_{a \in A})$ of {\normalfont \scshape Maximum Envy-Free Assignment} such that the following properties hold:
\begin{enumerate}
\item If $\opt_{\emph{MBB}}(G) \geq k$, then there exists an assignment $\phi^*$ such that $\val(\phi^*) \geq k$.
\item Given any assignment $\phi$ such that $\val(\phi) \geq k$, there is a polynomial-time algorithm that outputs $S \subseteq L$ and $T \subseteq R$ such that $|S|=|T|=\lfloor k/2\rfloor$, and $S$ and $T$ together induce a biclique in $G$.
\item $|A| = |L|$.
\end{enumerate}
\end{lemma}

Before we describe the reduction, let us explain how we can use it to prove Theorem~\ref{thm:max-ef-assignment-hardness}.

\begin{proof}[Proof of Theorem~\ref{thm:max-ef-assignment-hardness}]
Let $\varepsilon > 0$ be any constant, and suppose that there exists a polynomial-time $f(n)$-approximation algorithm $\cA$ for the {\normalfont \scshape Maximum Envy-Free Assignment} problem. 
We can use it to approximate {\normalfont \scshape Maximum Balanced Biclique} on input $G$ as follows:
\begin{itemize}
\item Run the reduction from Lemma~\ref{lem:thm:max-ef-assignment-reduction} to produce an instance $(A, H, \{u_a\}_{a \in A})$ of {\normalfont \scshape Maximum Envy-Free Assignment}.
\item Run $\cA$ on $(A, H, \{u_a\}_{a \in A})$ to get an assignment $\phi$.
\item Run the algorithm described in the second property of the reduction in Lemma~\ref{lem:thm:max-ef-assignment-reduction} on $\phi$ to get a balanced biclique $(S, T)$ in $G$.
\item Let $\beta := 2\left(\frac{1}{\varepsilon} + 1\right)$. Use a brute-force $(|L|+|R|)^{O(\beta)}$ algorithm to enumerate through all subsets of size at most $2\beta$, and consider the largest balanced biclique found.
\item Output the larger of the two bicliques computed in the previous two steps.
\end{itemize}

If $\opt_{\mbb}(G) / f(N) \leq \beta$, then the brute-force step of the algorithm ensures that the output biclique has size at least $\opt_{\mbb}(G) / f(N)$. As a result, we may henceforth assume that $\opt_{\mbb}(G) > f(N) \cdot \beta$. 

Now, from the first property of the reduction, there exists $\phi^*$ such that $\val(\phi^*) \geq \opt_{\mbb}(G)$. Thus, $\cA$ must output $\phi$ satisfying $\val(\phi) \geq \opt_{\mbb}(G) / f(n)$, which is equal to $\opt_{\mbb}(G) / f(N)$ due to the third property of the reduction. Then, the second property of the reduction ensures that our algorithm outputs a balanced biclique $(S,T)$ satisfying
\begin{align*}
|S| = |T| 
= \left\lfloor \frac{\opt_{\mbb}(G)}{2f(N)} \right\rfloor 
&> \frac{\opt_{\mbb}(G)}{2f(N)} - 1 \\
&> \frac{\opt_{\mbb}(G)}{2f(N)} - \frac{\opt_{\mbb}(G)}{f(N) \cdot \beta} = \frac{\opt_{\mbb}(G)}{2f(N) \cdot(1 + \varepsilon)},
\end{align*}
where the second inequality follows from $\opt_{\mbb}(G) > f(N) \cdot \beta$ and the last equality follows from our choice of $\beta$. 
It follows that our algorithm achieves an approximation ratio of $2f(N) \cdot(1 + \varepsilon)$ for {\normalfont \scshape Maximum Balanced Biclique}, as desired.
\end{proof}

To establish Theorem~\ref{thm:max-ef-assignment-hardness}, it therefore remains to prove Lemma~\ref{lem:thm:max-ef-assignment-reduction}.

\begin{proof}[Proof of Lemma~\ref{lem:thm:max-ef-assignment-reduction}]
Given an instance $G = (L, R, E)$ of {\normalfont \scshape Maximum Balanced Biclique} where $L = \{b_1, \dots, b_N\}$ and $R = \{c_1, \dots, c_M\}$, we create one agent $a_i$ for each vertex $b_i\in L$ and one house $h_j$ for each vertex $c_j\in R$. Moreover, we create $N$ additional houses $h^*_1, \dots, h^*_N$. 
(So, in total, there are $N$ agents and $M+N$ houses.)
The utility of each agent $a_i$ is defined by
\begin{align*}
u_{a_i}(h_j) =
\begin{cases}
N + j &\text{ if } (b_i, c_j) \notin E; \\
N &\text{ if } (b_i, c_j) \in E \\
\end{cases}
\end{align*}
for all $j \in [M]$, and
\begin{align*}
u_{a_i}(h^*_j) = j - 1.
\end{align*}
for all $j \in [N]$.

This completes the description of the reduction. It is clear that the reduction runs in polynomial time, and that the third property of the reduction holds. We will now prove the first two properties of the reduction.

\begin{enumerate}
\item Suppose that $\opt_{\mbb}(G) \geq k$, i.e., there exists a balanced biclique in $G$ where each side has $k$ vertices. 
Assume that this biclique consists of the vertices $b_{i_1}, \dots, b_{i_k}$ and $c_{i'_1}, \dots, c_{i'_k}$. Let us consider the following assignment:
\begin{align*}
\phi^*(a_i) =
\begin{cases}
h_{i'_\ell} &\text{ if } i = i_\ell \text{ for some } \ell \in [k]; \\
h^*_i &\text{ otherwise.}
\end{cases}
\end{align*}
Notice that each of $a_{i_1}, \dots, a_{i_k}$ has value $N$ for her own house, and does not value any assigned house more than $N$. As such, the assignment is envy-free for these $k$ agents, so $\val(\phi^*) \geq k$.
\item Suppose that there exists an assignment $\phi$ such that $\val(\phi) \geq k$. We may assume that $k \geq 2$, as otherwise we can simply output $S=T=\emptyset$.

Let $A_{\ef}$ denote the set of agents that are envy-free with respect to $\phi$, so $|A_{\ef}|\ge k$.  We start by showing that $\phi(A_{\ef}) \cap \{h^*_1, \dots, h^*_N\} = \emptyset$. 
Suppose for the sake of contradiction that for some $a \in A_{\ef}$, we have $\phi(a) = h^*_j$ for some $j \in [N]$.
Let $a'$ be another agent in $A_{\ef}$. 
Consider two cases based on whether $\phi(a') \in \{h^*_1, \dots, h^*_N\}$.
\begin{itemize}
\item[Case I:] $\phi(a') = h_{j'}$ for some $j' \in [M]$. In this case, we have $u_{a}(h_{j'}) \geq N > u_{a}(h^*_j)$, so the assignment is not envy-free for $a$.
\item[Case II:] $\phi(a') = h^*_{j'}$ for some $j' \in [N]$. In this case, if $j < j'$, then $a$ would envy $a'$; otherwise, if $j > j'$, then $a'$ would envy $a$.
\end{itemize}

In both cases, we have reached a contradiction, meaning that  $\phi(A_{\ef}) \cap \{h^*_1, \dots, h^*_N\} = \emptyset$. In other words, $\phi(A_{\ef}) \subseteq \{h_1, \dots, h_M\}$. 
Let $k$ elements of $\phi(A_{\ef})$ be $h_{j_1}, \dots, h_{j_k}$ where $j_1 < \cdots < j_k$, and let $a_{i_\ell} = \phi^{-1}(h_{j_\ell})$ for all $\ell \in [k]$. 
Let $S = \{b_{i_1}, \dots, b_{i_{\lfloor k/2 \rfloor}}\}$ and $T = \{c_{j_{k - \lfloor k/2 \rfloor + 1}}, \dots, c_{j_k}\}$, so $|S| = |T| = \lfloor k/2\rfloor$.

It remains to show that $S$ and $T$ together induce a biclique in $G$. 
Consider any $b_{i_\ell} \in S$ and $c_{j_{\ell'}} \in T$; notice that by our choice of $S$ and $T$, it holds that $\ell' > \ell$. 
Now, we have
\begin{align*}
u_{a_{i_\ell}}(\phi(a_{i_\ell})) = u_{a_{i_\ell}}(h_{j_\ell}) \leq N + j_\ell < N + j_{\ell'}.
\end{align*}
Since the house $h_{j_{\ell'}}$ is assigned and agent $a_{i_\ell}$ has utility strictly less than $N+j_{\ell'}$ for her assigned house, in order for her not to envy the owner of house $h_{j_{\ell'}}$, we must have $(b_{i_\ell}, c_{j_{\ell'}}) \in E$. 
\end{enumerate}
Hence, our reduction satisfies the claimed properties.
\end{proof}

Next, we show that even if the agents have binary utilities, maximizing the number of envy-free agents remains computationally hard.
This hardness only relies on the standard assumption P $\ne$ NP.

\begin{theorem}
\label{thm:EF-NP}
The problem of determining whether for a given positive integer $k$, 
there exists an assignment $\phi$ such that 
$\val(\phi) \geq k$, 
is NP-complete even when 
all agents have binary utilities. 
\end{theorem}

\begin{proof}
Since for each assignment $\phi$, 
we can compute $\val(\phi)$ in polynomial time, 
our problem is in NP. 
We prove the NP-hardness of our problem 
by reducing from the decision version of 
{\sc Minimum Coverage}. 
In this problem, we are given a finite set $E$ of elements, 
subsets $S_1,S_2,\dots,S_d$ of $E$, 
and positive integers $q, \ell$ such that 
$q \leq |E|$ and $\ell \leq d$; 
the goal is to determine whether there exists 
a subset $I \subseteq [d]$
such that 
$|I| = \ell$ and $|\bigcup_{t \in I}S_t| \leq q$. 
It is known that 
this problem is NP-complete \citep{Vinterbo02}.

Suppose that we are given an instance 
of 
the decision version of 
{\sc Minimum Coverage}. 
Then we construct an instance of our house allocation problem as follows. 

\begin{itemize}
\item
Define 
$A := \{a_e\mid e\in E\} \cup \{a^*_t \mid t \in [d]\}$ 
and $H := \{h^*_t\mid t\in [d]\} \cup 
\{h_t \mid t \in [|E|+ d - \ell]\}$.
\item 
For each element $e \in E$, define the utility 
function $u_{a_e} \colon H \to \{0,1\}$ by 
\begin{equation*}
u_{a_e}(h) := 
\begin{cases}
1 & \mbox{if $h = h^*_t$ such that $e \in S_t$}; \\ 
0 & \mbox{otherwise}. 
\end{cases}
\end{equation*}
\item 
For each integer $t \in [d]$, define the utility 
function $u_{a^*_t} \colon H \to \{0,1\}$ by 
\begin{equation*}
u_{a^*_t}(h) := 
\begin{cases}
1 & \mbox{if $h = h^*_t$}; \\ 
0 & \mbox{otherwise}. 
\end{cases}
\end{equation*}
\item 
Define $k := |E| + d - q$. 
\end{itemize}

($\Rightarrow$) Assume first that there exists a feasible solution $I \subseteq [d]$
to the decision version of {\sc Minimum Coverage}.
We will show that there exists a feasible solution to our house allocation problem. 
Define the assignment $\phi$ as follows. 
\begin{itemize}
\item 
For each integer $t \in I$, let
$\phi(a^*_t) := h^*_t$. 
\item Let $\pi$ be an arbitrary bijection from 
$A \setminus \{a^*_t \mid t \in I\}$ to 
$[|E| + d - \ell]$. 
For each agent $a \in A \setminus \{a^*_t \mid t \in I\}$, 
let $\phi(a) := h_{\pi(a)}$.

\end{itemize}
We claim that 
$\val(\phi) \geq k = |E| + d - q$.  
More precisely, we prove that 
every agent in 
\begin{equation*}
X := \left\{a_e \,\middle|\, e\in \left(E \setminus \bigcup_{t \in I} S_t\right)\right\}
\cup \{a^*_t \mid t \in [d]\}
\end{equation*}
is envy-free in $\phi$; notice that this is sufficient since $|A \setminus X| = |\{a_e\mid e\in \bigcup_{t \in I} S_t\}| \leq q$. 
\begin{itemize}
\item
Let $e$ be an element in $E \setminus \bigcup_{t \in I} S_t$. 
Then since
house $h^*_t$ is unassigned 
for every integer $t \in [d]$ such that
$e \in S_t$, $e$ has value $0$ for all assigned houses, and so $e$ is envy-free in $\phi$. 
\item 
Let $t$ be an integer in $[d]$. 
If $t \in I$, then since 
$u_{a^*_t}(\phi(a^*_t)) = u_{a^*_t}(h^*_t) = 1$, 
$a^*_t$ is envy-free in $\phi$. 
Else, $t \notin I$, and since 
house $h^*_t$ is unassigned, 
$a^*_t$ is again envy-free in $\phi$. 
\end{itemize}
This completes the proof of this direction.

($\Leftarrow$) Next, we prove the opposite direction.
That is, we assume that 
there exists an assignment $\phi$ such that 
$\val(\phi) \geq k$. 
We first prove that in this case, there exists 
an assignment $\sigma$ satisfying the following 
conditions. 
\begin{description}
\item[(A1)]
$\val(\sigma) \geq k$. 
\item[(A2)]
For every integer $t \in [d]$, if 
$\sigma^{-1}(h^*_t) \neq \emptyset$, then 
$\sigma^{-1}(h^*_t) = \{a^*_t\}$. 
\end{description} 

\begin{claim}
\label{claim}
There exists an assignment $\sigma$ satisfying 
{\normalfont (A1)} and {\normalfont (A2)}. 
\end{claim}

\begin{proof}[Proof of Claim~\ref{claim}]
Assume that there exists an 
integer $t \in [d]$ such that 
$\phi^{-1}(h^*_t) \neq \emptyset$ and 
$\phi^{-1}(h^*_t) \neq \{a^*_t\}$. 
Let $\widehat{a}$ be the agent in $A$ such that 
$\phi(\widehat{a}) = h^*_t$. 
Since $\phi(a^*_t) \neq h^*_t$, we have
$u_{a^*_t}(\phi(a^*_t)) = 0$
and
$u_{a^*_t}(h^*_t) = 1$, which implies that 
$a^*_t$ is not envy-free in $\phi$.

Define the assignment $\psi$ by 
\begin{equation*}
\psi(a) := 
\begin{cases}
h^*_t & \mbox{if $a = a^*_t$}; \\
\phi(a^*_t) & \mbox{if $a = \widehat{a}$}; \\
\phi(a) & \mbox{otherwise}. 
\end{cases}
\end{equation*}
Notice that the set of assigned houses in $\psi$ remains the same as in $\phi$.
This implies that 
for every agent $a \in A \setminus \{a^*_t, \widehat{a}\}$, 
$a$ is envy-free in $\phi$ if and only if 
$a$ is envy-free in $\psi$. 
Furthermore, 
since 
$u_{a^*_t}(\psi(a^*_t)) = u_{a^*_t}(h^*_t) = 1$, 
$a^*_t$ is envy-free in $\psi$. 
Thus, we have $\val(\psi) \geq \val(\phi) \geq k$. 
By setting $\phi := \psi$ and 
repeating this procedure, 
we eventually obtain a desired assignment $\sigma$.
\end{proof} 

Let $\sigma$ be an assignment satisfying 
(A1) and (A2) according to Claim~\ref{claim}. 
Define $I$ as the set of integers $t \in [d]$ such that 
$\sigma(a^*_t) = h^*_t$. 
Notice that 
\begin{equation*}
|I| \geq 
|A| - |\{h_t \mid t \in [|E|+d -\ell]\}| =  
|A| - (|E| + d - \ell) = \ell. 
\end{equation*}
For every element $e \in E$, if 
$e \in \bigcup_{t \in I}S_t$, then 
since there exists an integer $t \in [d]$ 
such that $e \in S_t$ and 
$\sigma(a^*_t) = h^*_t$, 
$a_e$ is not envy-free in $\sigma$. 
Thus, since $\val(\sigma) \geq k = |E| + d - q$, we have 
$|\bigcup_{t \in I}S_t| \leq q$. 
This completes the proof. 
\end{proof}

We remark that the {\sc Minimum Coverage} problem---also referred to as {\sc Bipartite Expansion}---is known to be hard to approximate~\citep{LouisRaVe13,KhotSa16}. However, since our reduction in Theorem~\ref{thm:EF-NP} is not approximation-preserving, it does not directly translate into a hardness of approximation for {\normalfont \scshape Maximum Envy-Free Assignment} in the case of binary utilities.

\section{Proportionality}

In this section, we address proportionality.
We show that deciding the existence of a proportional assignment is already NP-hard.
This is in contrast to envy-freeness, where \citet{GanSuVo19} gave an efficient algorithm for deciding whether an envy-free assignment exists.

\begin{theorem}
\label{thm:prop}
Deciding whether a proportional assignment exists in any given instance is NP-complete.
\end{theorem}

\begin{proof}
Membership in NP is clear: given an assignment, we can verify in polynomial time whether it is proportional.
For the hardness, we reduce from the {\normalfont \scshape Exact 3-Set Cover (X3C)} problem, where we are given a universe $V = \{v_1, \dots, v_N\}$ and subsets $S_1, \dots, S_M \subseteq V$, each of size $3$; the goal is to determine whether there exists a set cover of size $k := N / 3$. 
This problem is known to be NP-complete~\citep{GareyJo79}. 

Given an instance of {\normalfont \scshape X3C}, we perform the following reduction. For convenience, we will think of each subset $S_i$ as having ordered elements $e^0_i, e^1_i, e^2_i$.

We create $N$ agents $a_1, \dots, a_N$, each corresponding to an element in the universe, and $T := 100(M + N)$ additional agents $a^*_1, \dots, a^*_T$. Similarly, we create $3M$ houses $h^0_1, h^1_1, h^2_1, \dots, h^0_M, h^1_M, h^2_M$, where $h^0_j, h^1_j, h^2_j$ correspond to the subset $S_j$ for $j\in[M]$, as well as $T$ additional houses $h^*_1, \dots, h^*_T$. 
(So there are $N+T$ agents and $3M+T$ houses in total.)
Let $C := 8T + 8N - 19$.
The utilities for each of the first $N$ agents are defined by
\begin{align*}
u_{a_i}(h_j^\ell) =
\begin{cases}
8T & \text{ if } v_i = e^\ell_j; \\
6T - j & \text{ if } v_i = e^{\ell + 1\Mod 3}_j; \\
5T + j & \text{ if } v_i = e^{\ell + 2\Mod 3}_j; \\
0 & \text{otherwise}
\end{cases}
\end{align*}
and
\begin{align*}
u_{a_i}(h^*_j) = C.
\end{align*}
For the last $T$ agents, their utilities are defined by $u_{a^*_i}(h^\ell_j) = 0$ for all $j,\ell$ and $u_{a^*_i}(h^*_j) = 1$ for all $j$.

This completes the description of the reduction. 
It is clear that the reduction runs in polynomial time and that each utility value can be represented in $O(\log(NM))$ bits. 
We now establish the validity of the reduction.

($\Rightarrow$)
Suppose that there exist $k$ subsets $S_{p_1}, \dots, S_{p_k}$ that cover the entire universe. We may pick the following assignment: for the last $T$ agents, let $\phi(a^*_i) = h^*_i$ for all $i\in [T]$. Then, for each agent $a_i$ with $i\in[N]$, let $\phi(a_i) = h_j^\ell$, where $v_i = e_j^\ell$ belongs to the set cover.

Proportionality is clearly satisfied for the last $T$ agents. 
Furthermore, for each $i\in[N]$, agent~$a_i$ has utility exactly $8T$ for her assigned house, whereas $a_i$'s total utility for the $N+T$ assigned houses is
\begin{align*}
TC + 8T + (6T - j) + (5T + j) = 8T(N + T).
\end{align*}
Hence, proportionality is also satisfied for $a_i$.

($\Leftarrow$)
Suppose that there exists a proportional assignment $\phi$. We will show that the starting instance of {\normalfont \scshape X3C} is a YES instance.

To this end, let us first observe that since $N + T > 3M$, at least one of the houses $h^*_i$ must be assigned in $\phi$. 
From this and the assumption that $\phi$ is proportional for $a^*_1, \dots, a^*_T$, we must have 
\begin{align} \label{eq:dummy-get-dummy}
\phi(\{a^*_1, \dots, a^*_T\}) = \{h^*_1, \dots, h^*_T\}.
\end{align}

Let us denote $P := \phi(\{a_1, \dots, a_N, a^*_1, \dots, a^*_T\})$ and $Q := \phi(\{a_1, \dots, a_N\})$.
Observe that for each $i\in[N]$, \eqref{eq:dummy-get-dummy} implies that
\begin{align*}
\sum_{h\in P}u_{a_i}(h) = TC + \sum_{h\in Q}u_{a_i}(h).
\end{align*}
Furthermore, since $\phi$ is proportional for $a_i$, we have
\begin{align} \label{eq:prop}
u_{a_i}(\phi(a_i)) \geq \frac{TC + \sum_{h\in Q}u_{a_i}(h)}{N + T}.
\end{align}
This implies that $u_{a_i}(\phi(a_i)) \geq \frac{TC}{N + T} > 6T$, where the latter inequality follows from our choice of parameters. As a result, we must have 
\begin{align} \label{eq:assigned-item-util}
u_{a_i}(\phi(a_i)) = 8T.
\end{align} 
Equivalently, this can be stated as 
\begin{align*}
\phi(a_i) = h_{j_i}^{\ell_i} \text{ where } e_{j_i}^{\ell_i} = v_i.
\end{align*}
Next, notice that 
\begin{align*}
\sum_{i \in [N]} \sum_{h\in Q}u_{a_i}(h) &= \sum_{i, i' \in [N]} u_{a_i}(h^{\ell_{i'}}_{j_{i'}}) = \sum_{i' \in [N]} \sum_{i \in [N]} u_{a_i}(h^{\ell_{i'}}_{j_{i'}}) = \sum_{i' \in [N]} 19T = 19TN.
\end{align*}
Summing \eqref{eq:prop} over $i\in [N]$ and plugging in the above relation, we get
\begin{align*}
\sum_{i \in [N]} u_{a_i}(\phi(a_i)) \geq \frac{TCN + 19TN}{N + T} = 8NT.
\end{align*}
From~\eqref{eq:assigned-item-util}, this inequality is an equality and, as a result,~\eqref{eq:prop} must be an equality for all $i\in[N]$ as well. This implies that
\[ 
\sum_{h\in Q}u_{a_i}(h) = 19T
\]
for all $i \in [N]$.

Finally, observe that the only way for each agent  $a_i$ to get a utility of $19T$ from the houses in $Q$ is through $8T, 6T - j$ and $5T + j$ for some $j \in [M]$. From this, we can conclude that if $h_j^\ell$ belongs to $Q$ for some $j, \ell$, then it must be that all of $h_j^0, h_j^1, h_j^2$ belong to $Q$. As a result, $Q$ corresponds to $N/3 = k$ subsets in the {\normalfont \scshape X3C} instance. 
Furthermore, the fact that $u_{a_i}(\phi(a_i)) = 8T$ for all $i \in [N]$ ensures that these subsets form a set cover of the universe. This concludes our proof.
\end{proof}

We remark that the difficulty of deciding the existence of a proportional allocation stems from the fact that unallocated houses are not taken into account in the definition of proportionality.
In particular, if we were to use an alternative definition wherein each agent calculates her proportional share based on her utility for the set of \emph{all} houses, the problem would become solvable in polynomial time, since we would know the desired threshold for every agent and could then check whether a proportional allocation exists by matching.

For binary utilities, envy-freeness and proportionality are equivalent.
Indeed, if an agent has utility $0$ for all $n$ assigned houses, then she is both envy-free and proportional, while if the agent has utility $1$ for at least one assigned house, then envy-freeness and proportionality are both equivalent to the condition that the agent receives a house for which she has utility $1$.
Theorem~\ref{thm:EF-NP} therefore implies the following corollary.

\begin{corollary}
The problem of determining whether for a given positive integer $k$, 
there exists an assignment such that at least $k$ agents are proportional, 
is NP-complete even when 
all agents have binary utilities. 
\end{corollary}

\section{Equitability}

Finally, we turn our attention to equitability and show that in contrast to  proportionality, deciding whether an equitable assignment exists can be done efficiently.

\begin{theorem}
\label{thm:equitable}
There is a polynomial-time algorithm that, for any given instance, decides whether an equitable allocation exists.
\end{theorem}

\begin{proof}
We iterate over the values $u_{a_i}(h_j)$ for all $i\in[n]$ and $j\in [m]$.
For such each value $k$, we construct a bipartite graph  $G=(A,H,E)$, where there is an edge between agent $a_i$ and house $h_j$ if and only if $u_{a_i}(h_j) = k$, and compute a maximum matching of the graph.
We return that an equitable assignment exists exactly when the maximum matching has size $n$ for at least one constructed graph.

It is well-known that computing a maximum matching in a bipartite graph can be done in polynomial time, and the number of values $u_{a_i}(h_j)$ is $O(mn)$.
If we find a matching of size $n$, this clearly corresponds to an equitable assignment.
Conversely, if there is an equitable assignment with value $k$, then the assignment gives rise to a matching of size $n$ in the bipartite graph constructed for value $k$.
\end{proof}

\section{Concluding Remarks}

In this paper, we have studied the complexity of computing fair house allocations with respect to envy-freeness, proportionality, and equitability.
We conclude with some questions that remain from our work.
\begin{itemize}
\item What is the best approximation ratio for maximizing the number of envy-free agents under binary utilities in polynomial time?
\item What is the complexity of deciding whether a proportional assignment exists under binary utilities?
\item Define the \emph{inequity} of an assignment $\phi$ as the difference between the highest and lowest utilities in $\phi$. 
What is the complexity of computing an assignment with the smallest inequity?
\end{itemize}

\subsection*{Acknowledgments}

This work was partially supported by JSPS KAKENHI Grant Number
JP20H05795, Japan and by an NUS Start-up Grant.
We thank the anonymous reviewer for valuable feedback.

\bibliographystyle{plainnat}
\bibliography{main}

\end{document}